\documentclass[11pt]{article}

\usepackage{enumerate}

\usepackage[utf8]{inputenc}
\usepackage[english]{babel}
\usepackage{graphicx}         
\usepackage{braket}
\usepackage{amsmath}
\usepackage{amssymb}
\usepackage{amsthm}
\usepackage{mathtools}
\usepackage{bbm}		
\usepackage{mathrsfs}
\usepackage{amssymb}
\usepackage{array}		
\usepackage{color}
\usepackage{subeqnarray}	
\usepackage{wrapfig}
\usepackage{url}
\usepackage{hyperref}
\usepackage{physics}
\usepackage{tikz}
\usepackage{verbatim}
\usetikzlibrary{angles,quotes}
\usetikzlibrary {arrows.meta}

\usepackage{bbm}

\usepackage[backend=bibtex,
style=numeric,
sorting=nyt,
doi=false,
isbn=false,
url=false,
eprint=false,
giveninits=true,
maxnames=99
]{biblatex}
\AtEveryBibitem{%
	\clearfield{doi}%
	\clearfield{url}%
	\clearfield{eprint}%
	\clearfield{isbn}%
	\clearfield{month}%
}

\renewbibmacro*{journal+issuetitle}{%
	\usebibmacro{journal}%
	\setunit*{\addspace}%
	\printfield{volume}%
	\iffieldundef{number}
	{}
	{\mkbibparens{\printfield{number}}}%
	\setunit{\addcomma\space}%
	\printfield{eid}%
	\setunit{\addcomma\space}%
	\usebibmacro{issue+date}%
	\newunit}

\allowdisplaybreaks

\numberwithin{equation}{section}

\newtheorem{theorem}{Theorem}[section]
\newtheorem{lemma}{Lemma}[section]

\newtheorem{corollary}{Corollary}[section]

\newtheorem{assumption}{Assumption}[section]
\newenvironment{assumptionA}[1]
{%
	\begingroup
	\renewcommand{\theassumption}{(#1)}%
	\begin{assumption}%
	}
	{%
	\end{assumption}%
	\endgroup
	\addtocounter{assumption}{-1}%
}

\title{On the Emergence of Exponential Decay from Discrete Spectra for Friedrichs Hamiltonians}
\author{
	Tamim Al-Qaiwani\thanks{School of Science, Constructor University Bremen, Campus Ring 1, 28759 Bremen, Germany. \texttt{talqaiwani@constructor.university}}\, and Sören Petrat\thanks{School of Science, Constructor University Bremen, Campus Ring 1, 28759 Bremen, Germany. \texttt{spetrat@constructor.university}}}

\date{\today}

\begin{document}
	\maketitle

	\begin{abstract}
		We study a class of Friedrichs Hamiltonians, that is, operators describing an excited state coupled to a bath through a rank-one perturbation, in the case where the bath Hamiltonian has purely discrete spectrum. We consider sequences of such Hamiltonians for which the spectral measures associated to the coupling functions by the bath Hamiltonians converge weakly to a limiting measure that is absolutely continuous with a Hölder continuous density near the excited energy. Under this assumption, we show that the survival probability of the excited state decays in an approximately exponential manner on suitable time scales and under favourable conditions, with a decay rate given by Fermi’s golden rule for the limiting measure. The error is estimated in terms of the coupling strength and the Lévy distance between the discrete spectral measures and the limiting measure. As an application, we treat a two-level atom coupled to a massless bosonic field confined to a large cavity in the rotating-wave approximation.
	\end{abstract}
\section{Introduction and Main Results}
There are many instances of quantum systems where an initially stable state becomes metastable when a perturbation is switched on. A prototypical example are the excited states of atomic systems. Since they are eigenstates of the atomic Hamiltonian, the atom will remain in any given excited state for an indefinite amount of time once prepared in such a state. However, as soon as a coupling to a photon field is introduced, the excited states (but not the ground state) should become metastable and the probability of the system remaining in a given excited state should decay in an approximately exponential manner on suitable timescales \cite{OkamotoYajima1985, King1994, BachFröhlichSigal,HaslerHerbstHuber2008,BACH2017,}. \\

\noindent
From an abstract mathematical perspective, the setting is the following: One considers a self-adjoint operator $H_0$ on a separable Hilbert space $\mathcal{H}$ describing the unperturbed system and a normalised eigenstate $\chi\in \mathcal{H}$ of $H_0$, describing the initially stable state. For simplicity we shall assume that the corresponding eigenvalue $E_0$ is non-degenerate. The perturbation is typically assumed to be a symmetric $H_0$-bounded operator $W$, so that $H(\mu)=H_0+\mu W$ is self-adjoint on $D(H_0)$ for sufficiently small $\mu\geq 0$. The quantity we shall use as an indicator for metastability is the \textit{survival probability}, defined as
\begin{equation}
	p^\mu(t)=\abs{\left\langle \chi, \exp(-itH(\mu))\chi\right\rangle}^2.
\end{equation}
This is the probability for a positive answer to the following question: ``Prepare the system in the state $\chi$ and consequently evolve for a time $t\geq 0$, where the dynamics is governed by the Hamiltonian $H(\mu)$. Is the system still in the state $\chi$ at time $t$?" In the case where the eigenvalue $E_0$ is surrounded by the absolutely continuous spectrum of $H_0$, there is a large body of literature devoted to proving that the survival probability decays in an approximately exponential manner in time, in the sense that there is a $\mu_0> 0$ and a function $\Gamma:(0,\mu_0)\to(0,\infty)$, such that
\begin{equation}
	\label{Eq:DecayLaw}
	\lim_{\mu\to 0^+}\sup_{t\geq 0}\abs{p^\mu(t)-e^{-2\Gamma(\mu)t}}=0.
\end{equation}
To the best of our knowledge, however, there are no works which address the case where the eigenvalue $E_0$ is surrounded, not by absolutely continuous, but by discrete, albeit ``dense", spectrum. Nor do any of the methods used in the case where the eigenvalue $E_0$ is embedded in the absolutely continuous spectrum of $H_0$ apply, in any immediate way, to this setting. Most notably, these methods include the complex deformation method of Aguilar--Balslev--Combes \cite{AguilarCombes,BalslevCombes,SimonNBodyResonances} and methods bases on limiting absorption principles in combination with Mourre theory \cite{Mourre,JensenMourrePerry,OrthResonances,Cattaneo2005AGR}. The situation of discrete, yet ``dense", spectrum, is, however, often encountered in statistical mechanics, where the particles under investigation are habitually confined to a large but finite volume, say to a torus of side length $L>0$. This leads to a discretisation of the energy of the system, which becomes ``denser" the larger the confinement is.
The present article is concerned with addressing precisely such situations, at least in the, admittedly rather restricted, case of so-called Friedrichs Hamiltonians. These Hamiltonians, named after the seminal work of Friedrichs \cite{Friedrichs} on perturbations of continuous spectra, can be specified by the following data:
\begin{enumerate}
	\item The energy $E_0\in\mathbb{R}$ of the excited state of the ``atom".
	\item A self-adjoint operator $K$ acting on a separable Hilbert space $\mathcal{K}$, describing the ``bath" to which the ``atom" is coupled.
	\item A coupling function (or form factor) $g\in\mathcal{K}$, which we will assume to be normalised, describing the form of the coupling of ``atom" to ``bath".
\end{enumerate}
For $\mu\geq 0$, the Friedrichs Hamiltonian associated to the triple $(E_0,K,\mu g)$ is then defined to be the self-adjoint operator $H(\mu)$ acting on $\mathcal{H}=\mathbb{C}\oplus\mathcal{K}$ as
$$
H(\mu )(\lambda, \psi)=(E_0\lambda + \mu\langle g,\psi\rangle , K\psi + \mu \lambda g),
$$
with domain $D(H(\mu))=\mathbb{C}\oplus D(K)$. Hamiltonians of this form are frequently used as approximations to various system in the physics literature \cite{FriedrichModelsinPhysics}. For aspects concerning the mathematical theory of Friedrichs Hamiltonians, we refer the reader to the lecture notes of Jakšić, Kritchevski and Pillet \cite{Jakšić2006} and the references therein. In particular, the authors present a proof for the exponential decay law (\ref{Eq:DecayLaw}) when $E_0$ is surrounded by the absolutely continuous spectrum of $K$, under certain additional analyticity assumptions, using the method of complex deformations. For questions regarding the definition of Friedrichs Hamiltonians when the coupling function is singular, that is, when $g\notin\mathcal{K}$, we refer the reader to the works of Derezi{\'n}ski and Früboes \cite{DEREZINSKI2002433} and Facci, Ligabò and Lonigro \cite{FacchiLigaboLonigro2019,SingularFriedrich}. In this article, however, we shall only consider regular coupling functions $g\in\mathcal{K}$.\\

\noindent
Let us now discuss a few examples of systems which fall into the setting we wish to address in this article, namely systems described, perhaps in an approximate sense, by Friedrichs Hamiltonians with a discrete, yet ``dense", spectrum.
The guiding example is provided by the confined spin-boson model in the rotating-wave approximation. This model describes a two-level system, which we shall refer to as the atom, coupled to a massless bosonic field confined to a torus of side-length $L>0$. The Hamiltonian describing the atom is 
	\begin{equation}
	H^{\text{Atom}}(E_0)=\begin{pmatrix}
		E_0  & 0\\
		0 & 0
	\end{pmatrix}	
\end{equation}
on $\mathbb{C}^2$, in the basis $|e\rangle =(1,0)$, $|g\rangle =(0,1)$, where $|e\rangle$ describes the excited state of the atom, $|g\rangle$ is the ground state and $E_0>0$ is their energy difference. Letting $\Lambda_L:=2\pi L^{-1} \mathbb{Z}^3$ denote the momentum lattice dual to the torus of side-length $L>0$, the Hamiltonian describing the field energy is defined on the bosonic Fock space $\mathcal{F}_b(\ell^2(\Lambda_L))=\bigoplus_{n=0}^\infty \ell^2_{sym}(\Lambda_L^n)$ as
\begin{equation}
	H^{\text{Field}}_L=\sum_{k\in\Lambda_L}|k|a_k^\dagger a_k,
\end{equation}
with the natural domain, which we abbreviate by $\mathcal{D}_L$.
Here, $a_k^\dagger$ and $a_k$ denote the creation and annihilation operators for a field quantum of momentum $k\in\Lambda_L$, satisfying the canonical commutation relations $[a_k,a_q^\dagger]=\delta_{kq}$ and $[a_k,a_q]=[a_k^\dagger,a_q^\dagger]=0$. It is clear that $H^{\text{Field}}_L$ has purely discrete spectrum given by  $\sigma(H^{\text{Field}}_L)=\{2\pi |k| L^{-1} \text{ : }k\in\mathbb{Z}^3\}$, which becomes ``denser" the larger the confinement becomes. The Hilbert space describing the combined system is $\mathbb{C}^2\otimes \mathcal{F}_b(\ell^2(\Lambda_L))$. For some $g_L\in \ell^2(\Lambda_L)$ with $\norm{g_L}_{\ell^2(\Lambda_L)}=1$ and $\mu\geq 0$, consider the spin-boson Hamiltonian
\begin{equation}
	\label{eq:Hamiltonian}
	H_L^{\text{SB}}(E_0,\mu g_L)=H^{\text{Atom}}(E_0)\otimes I+ I\otimes H^{\text{Field}}_L +\mu \sum_{k\in\Lambda_L}\left(\sigma_{-}\otimes g_L(k )a_k^\dagger + \sigma_+\otimes \overline{g_L(k)}a_k\right),
\end{equation}
where $\sigma_{-}=|g\rangle\langle e|$ and $\sigma_{+}=|e\rangle\langle g|$. Note that the above Hamiltonian is self-adjoint on $\mathbb{C}^2\otimes\mathcal{D}_L$ by the Kato--Rellich theorem, and that it only contains interaction terms which are capable of producing a field particle by relaxation of the excited state of the atom, but not by excitation of the ground state. The neglect of the latter terms is called the rotating-wave approximation. The total number of excitations
\begin{equation}
	\mathcal{N}_{e}=|e\rangle\langle e|\otimes I + I \otimes \sum_{k\in\Lambda_L}a_k^\dagger a_k
\end{equation}
commutes with $H_L^{\text{SB}}(E_0,\mu g_L)$ and is therefore a conserved quantity. Hence, the time evolution of the state describing the excited atom in the absence of field particles, namely $|e\rangle\otimes\Omega$, where $\Omega$ denotes the Fock space vacuum, can be computed using $H_L^{\text{SB}}(E_0,\mu g_L)$ restricted to the one-excitation sector. It is easily verified that this restricted Hamiltonian is unitarily equivalent to the Friedrichs Hamiltonian associated to the triple $(E_0,K_L,\mu g_L)$, where $K_L$ is the self-adjoint operator acting on $\ell^2(\Lambda_L)$ by $(K_Lf)(k)=|k|f(k)$, and that, under this identification, $|e\rangle\otimes \Omega$ is mapped to the state $\chi=(1,0)\in \mathbb{C}\otimes \ell^2(\Lambda_L)$. In the case where the field particles are allowed to propagate along the entire real line, which in particular implies that $E_0$ is surrounded by absolutely continuous spectrum, the survival probability of the state $|e\rangle\otimes \Omega$ was shown to decay exponentially in the rotating-wave approximation by King \cite{KingWW}, without analyticity assumptions. As an application of our main result, we appropriately extend this result to the setting of confined field particles in three dimensions in Section \ref{Sec:WWAtom}. Another example of a similar nature is given by translation invariant Nelson or Polaron type models in finite volume and where the bosonic Fock space is restricted to the vacuum and one-particle sector, such as those considered by G{\'e}rard, M{\o}ller and Rasmussen in \cite{GerardMollerRasmusses2011}. Here, the Friedrichs Hamiltonians arise as fiber Hamiltonians describing the system at a fixed total momentum and the discreteness of the spectrum stems again from the confinement of the field particles. A third example can be found in the work of Derezi{\'n}ski, Li, and Napi{\'o}rkowski \cite{DerezinskiNapiorkowskiLee2024}, who used a Friedrichs Hamiltonian to analyse the damping of quasi-particle excitations in a Bose gas confined to a large volume, called Beliaev damping. This example is of particular interest, since the $N$-body Hamiltonian describing the Bose gas is most naturally formulated in finite volume. After a series of transformations and approximations, which we do not reproduce here, the authors arrive at a Friedrichs Hamiltonian with purely discrete, yet ``dense", spectrum. In order to refer to known results in the setting of absolutely continuous spectrum, the authors then take a formal continuum limit, thereby replacing the discrete spectrum by an absolutely continuous one. We hope that the results in this paper may be used to treat the finite-volume Friedrichs Hamiltonians themselves and show that approximately exponential decay already emerges on suitable timescales before the thermodynamic limit is taken, provided the discrete spectrum is sufficiently dense. As a final example of a somewhat different nature we mention the Bixon--Jortner model \cite{BixonJortner}, which is used in quantum chemistry to describe the radiationless decay of a state into the quasi-continuous vibrational modes of a large molecule. We should remark that, strictly speaking, the Bixon--Jortner model does not fall into the class of Hamiltonians we consider here, since the ``atom" is assumed to be uniformly coupled to an infinitude of states, with energies ranging from negative to positive infinity, so that $g\notin\mathcal{K}$. However, this is an unphysical aspect of the model and our analysis applies to more realistic versions of this model, where the energy levels of the molecule lie in a band of finite width. \\

\noindent
Let us now discuss our assumptions on the model in more detail. In order to express that the spectrum of the Friedrichs Hamiltonians under consideration should be ``almost" continuous, we assume that, instead of one, we are given a sequence of such Hamiltonians $(H_L(\mu))_{L\in\mathbb{N}}$ associated to the triples $(E_0,K_L,\mu g_L)$. Letting $\nu_L$ denote the spectral measure associated to $g_L$ by $K_L$, we impose the following assumption:
\begin{assumptionA}{A$\alpha$}
	\label{Assumption}
	The spectral measures $\nu_L$ converge weakly to a measure $\nu_\infty$ which is absolutely continuous w.r.t.\ the Lebesgue measure on some compact neighbourhood $I$ of $E_0$. Moreover, the corresponding density function $\rho_\infty$ is bounded and Hölder continuous of degree $\alpha\in(0,1]$ on $I$. Finally, we assume that $\rho_\infty(E_0)>0$.
\end{assumptionA}
\noindent
In Lemma \ref{Lemm_Properties}, we verify that this assumption holds for sequences of confined spin-boson Hamiltonians in the rotating wave approximation $H_L^\text{SB}(E_0,\mu g_L)$ when restricted to the one-excitation sector and under certain regularity assumptions on the coupling functions $g_L\in\ell^2(\Lambda_L)$.\\

\noindent
In order to quantify how ``continuum like" the spectral measures $\nu_L$ are, we use the Lévy metric, defined on the space $\mathcal{P}(\mathbb{R})$ of probability measures over $(\mathbb{R},\mathcal{B}(\mathbb{R}))$ by
\begin{equation}
	d(P,Q):=\inf\{\varepsilon>0 : F_P(\lambda-\varepsilon) - \varepsilon\leq F_Q(\lambda)\leq F_P(\lambda+\varepsilon) + \varepsilon\text{ for all }\lambda\in\mathbb{R}\},
\end{equation}
where for a probability measure $P\in\mathcal{P}(\mathbb{R})$, we denote by $F_P(\lambda)=P((-\infty,\lambda])$ the corresponding cumulative distribution function. If $(H_L(\mu))_{L\in\mathbb{N}}$ is a sequence of Friedrichs Hamiltonians satisfying Assumption \ref{Assumption} for some $\alpha\in(0,1]$, then $d(\nu_L,\nu_\infty)\xrightarrow{L\to\infty}0$, since the Lévy metric metrises weak convergence \cite{Elstrodt}. With a slight abuse of notation, we will let $\chi=(1,0)\in\mathcal{H}_L\equiv \mathbb{C}\oplus \mathcal{K}_L$, neglecting the dependence on $L$. Our main result is:
\begin{theorem}
	\label{Theorem:ExpDecay}
	Let $(H_L(\mu))_{L\in\mathbb{N}}$ be a sequence of Friedrichs Hamiltonians associated to triples $(E_0,K_L,\mu g_L)$ satisfying Assumption \ref{Assumption} for some $\alpha\in(0,1]$. Let $\nu_\infty$ denote the weak limit of the spectral measures associated to $g_L$ by $K_L$ as $L$ approaches infinity, and define the decay rate $\Gamma(E_0)$ and the energy shift $\Delta(E_0)$ by
	\begin{align}
		\label{eq:EnergyShift}
		\Delta(E_0)&=\lim_{\delta\to 0^+}\int_{\mathbb{R}\setminus B_\delta(E_0)}\frac{1}{\lambda - E}\dd \nu_\infty(\lambda),\\
		\Gamma(E_0)&=\pi\frac{\dd \nu_\infty}{\dd\lambda}(E_0)\equiv \pi\rho_\infty(E_0)\label{eq:DecayRate}.
	\end{align}
Then, for any $\gamma\in (0,\alpha)$, there is a $C(\gamma)>0$ such that, for all $\mu\geq 0$ and $L\in\mathbb{N}$,
	\begin{equation}
		\left\langle \chi,\exp(-itH_L(\mu))\chi\right\rangle = e^{-i(E_0-\mu^2 \Delta(E_0))t}e^{-\mu^2\Gamma(E_0)t}+R^\mu_L(t),\quad t\geq 0,
	\end{equation}
	where the error term $R^{\mu}_L(t)$ satisfies the estimate
	\begin{equation}
		\abs{R^\mu_L(t)}\leq C(\gamma)\left(\mu^{2\gamma} + \left(\frac{d(\nu_L,\nu_\infty)}{\mu^{2}}\right)^{\alpha} \right)\exp(\frac{d(\nu_L,\nu_\infty)}{\mu^{2\gamma}}t).
	\end{equation}
\end{theorem}

\noindent
\textit{Remarks:}
\begin{enumerate}
	\item In the special case of a Friedrichs Hamiltonian $H(\mu)$ associated to a triple $(E_0,K,\mu g)$ where the spectral measure $\nu$ associated to $g$ by $K$ already has a Hölder continuous density of degree $\alpha\in(0,1]$ close to $E_0$, the above theorem recovers the uniform time estimate reported, for example, in \cite{JensenNenciuOddDim} by Jensen and Nenciu:
\begin{equation*}
	\sup_{t\geq 0}\abs{\left\langle \chi,\exp(-itH(\mu))\chi\right\rangle - e^{-i(E_0-\mu^2 \Delta(E_0))t}e^{-\mu^2\Gamma(E_0)t}}\leq C(\gamma)\mu^{2\gamma},\quad \gamma\in(0,\alpha).
\end{equation*}
\item The quantity $\mu^2/d(\nu_L,\nu_\infty)$ appearing in the error term can be roughly thought of as comparing two spectral length scales: The width of ``discrete resonance", which is of the order $\mu^2$, and the length scale determined by the discreteness of the measure $\nu_L$, which is quantified by $d(\nu_L,\nu_\infty)$. The reciprocal of this ratio being small can be used to give meaning to the phrase that the spectrum of $K$ around $E_0$ should be sufficiently ``dense".
\item In contrast to results obtained in the continuum setting, the error term in Theorem \ref{Theorem:ExpDecay} is not uniformly small in time, but grows exponentially. A growing error is necessary in the discrete setting due to the possibility of Poincaré recurrences, where the survival probability returns arbitrarily close to one after a sufficiently long, but finite, duration of time \cite{BocchieriLoinger}.
\item Under the assumptions of Theorem \ref{Theorem:ExpDecay}, the energy shift and decay rate can be computed using Fermi's golden rule for the limiting measure:
\begin{equation*}
	\Delta(E_0)+i\Gamma(E_0)=\lim_{\varepsilon\to 0^+}\lim_{L\to\infty}\langle g_L,(K_L-E_0-i\varepsilon)^{-1}g_L\rangle=\lim_{\varepsilon\to 0^+}\int_{\mathbb{R}}\frac{1}{\lambda-E_0-i\varepsilon}\dd\nu_\infty(\lambda).
\end{equation*}
This is proven in Lemma \ref{Lemma3}.
\end{enumerate}
The approximate exponential decay of the survival amplitude is often also expressed in terms of a weak coupling (or van-Hove, after \cite{VANHOVE1954517}) limit, where the limit of weak coupling $\mu\to 0^+$ is combined with a limit of  large times $t\to\infty$ in such a way that the quantity $\tau:=\mu^2 t$ is kept fixed.  A possible formulation in the discrete setting is the following immediate corollary, where the limit of weak coupling $\mu\to 0^+$ is additionally combined with a continuum limit $L\to\infty$.
\begin{corollary}
	\label{Corallary_SimWeakContLimit}
	Let $(H_L(\mu))_{L\in\mathbb{N}}$ be a sequence of Friedrichs Hamiltonians associated to triples $(E_0,K_L,\mu g_L)$ satisfying Assumption \ref{Assumption} for some $\alpha\in(0,1]$. Let $\beta\in(0,1/2)$ be arbitrary and put $\mu_L=d(\nu_L,\nu_\infty)^\beta$. Fix a $\tau\geq 0$. Then the simultaneous weak-coupling and continuum limit
	\begin{equation}
		a_{wcc}(\tau):=\lim_{L\to\infty}\left\langle \chi, \exp(-i\tau\mu_L^{-2}(H_L(\mu_L)-E_0))\chi\right\rangle,
	\end{equation}
	exists and is given by
	\begin{equation}
		a_{wcc}(\tau)=e^{i\Delta(E_0)\tau}e^{-\Gamma(E_0)\tau},
	\end{equation}
with $\Delta(E_0)$ and $\Gamma(E_0)$ defined in \eqref{eq:EnergyShift} and \eqref{eq:DecayRate}. Moreover, for any $0<\alpha^-<\alpha$, there is a $C_{\tau,\alpha^-}>0$ such that $$\abs{a_{wcc}(\tau)-\left\langle \chi, \exp(-i\tau\mu_L^{-2}(H_L(\mu_L)-E_0))\chi\right\rangle}\leq C_{\tau,\alpha^-} d(\nu_L,\nu_\infty)^{\min\{2\beta\alpha^-,(1-2\beta)\alpha\}}.$$
\end{corollary}

\noindent
In Theorem \ref{Theorem:ExpDecaySB} we apply the above corollary to the confined spin-boson model in the rotating-wave approximation. Moreover, the condition $\beta<1/2$ is essentially sharp. Indeed, in Theorem \ref{Theorem:ExpDecaySB} we further prove that any $\beta>1/2$ does not lead to an exponential decay of the survival probability of the state $|e\rangle\otimes\Omega$.\\

\noindent
We will now briefly outline the proof of Theorem \ref{Theorem:ExpDecay}. The starting point is the spectral representation of the survival amplitude
\begin{align}
	\langle\chi,\exp(-itH_L(\mu)t)\chi\rangle = \int_{-\infty}^\infty e^{-it\lambda}\dd\sigma^\mu_L(\lambda),
\end{align}
where $\sigma^\mu_L$ is the spectral measure associated to $\chi=(1,0)$ by $H_L(\mu)$. 
The measure $\sigma^0_L$ is simply a Dirac measure situated at $E_0$. If the perturbation is now increased to values significantly larger than the scale determined by the discreteness of $\nu_L$, one can picture the eigenvalues of $H_L(\mu)$ as mixing heavily with one-another. In the process, the initial mass of the Dirac peak will be spread to a large number of neighbouring eigenvalues. Thus, the form of the measure $\sigma^\mu_L$ is rather complicated and regular Kato--Rellich perturbation theory cannot be applied to obtain useful information. Instead of working with the measure $\sigma^\mu_L$ directly, we will use a smeared version, defined by the density
\begin{align}
	\label{eq:rho}
	\rho^\mu_L(\varepsilon;E):=(p_\varepsilon * \sigma^\mu_L)(E),
\end{align}
where $p_\varepsilon(\lambda):=\pi^{-1}\varepsilon (\lambda^2+\varepsilon^2)^{-1}$ denotes the Cauchy kernel of scale $\varepsilon>0$. Clearly, we have that $\rho_L^\mu(\varepsilon;\cdot)\in L^1(\mathbb{R})$ with $\norm{\rho_L^\mu(\varepsilon;\cdot)}_{L^1(\mathbb{R})}=1$. The scale $\varepsilon>0$ must be chosen large enough as to effectively smear out the discrete structure of the measure $\sigma_L^\mu$, while being sufficiently small as to still accurately resolve the structure of interest, which is a ``discrete resonance" of width $\mu^2$, i.e., $\varepsilon$ must be chosen such that $d(\nu_L,\nu_\infty)\ll \varepsilon\ll \mu^2$. By the Fourier convolution theorem, we have that
\begin{align}
	\label{eq:SmearFourier}
	\int_{-\infty}^\infty \rho_L^\mu(\varepsilon;\lambda)e^{-it\lambda}\dd\lambda = \widehat{p_\varepsilon}(t) \langle\chi,\exp(-itH_L(\mu)t)\chi\rangle,
\end{align}
where $\widehat{p_\varepsilon}(t)=e^{-\varepsilon |t|}$ denotes the Fourier transform of $p_\varepsilon$. In order to proceed, we then show that the density $\rho_L^\mu(\varepsilon;\cdot)$ is well approximated in the $L^1(\mathbb{R})$-norm by the density
\begin{align}
	\label{eq:rhotilde}
	\tilde{\rho}_L^\mu(\varepsilon;E)=\frac{1}{\pi}\frac{(\varepsilon + \mu^2\Gamma(E_0))}{(E-E_0+\mu^2\Delta(E_0))^2+(\varepsilon + \mu^2\Gamma(E_0))^2},
\end{align}
where $\Delta(E_0)$ and $\Gamma(E_0)$ are defined by (\ref{eq:EnergyShift}) and (\ref{eq:DecayRate}). Since for any $t\geq 0$, the Fourier transform of $\tilde{\rho}_L^\mu(\varepsilon;\cdot)$ is given by $\widehat{\tilde{\rho}_L^\mu(\varepsilon;\cdot)}(t)=e^{-i(E_0-\mu^2\Delta(E_0))t}e^{-\mu^2\Gamma(E_0)t}e^{-\varepsilon t}$, Equation (\ref{eq:SmearFourier}) then yields that, for any $t\geq 0$,
\begin{align}
	\label{eq:3}
	\abs{\langle\chi,\exp(-itH_L(\mu)t)\chi\rangle-e^{-i(E_0-\mu^2\Delta(E_0))t}e^{-\mu^2\Gamma(E_0)t}}\leq \norm{\rho_L^\mu(\varepsilon;\cdot)-\tilde{\rho}_L^\mu(\varepsilon;\cdot)}_{L^1(\mathbb{R})}e^{\varepsilon t}.
\end{align}
The remaining difficulty, in order to conclude the proof of Theorem \ref{Theorem:ExpDecay}, is to estimate the $L^1(\mathbb{R})$-norm appearing above after appropriately choosing the scale $\varepsilon>0$.
\\

\noindent
The rest of this paper is structured as follows: In Section \ref{Sec:SelfEnergyReg} we prove two regularity results on the self-energies of the model. In Section \ref{Sec:ExpDecay}, these results are then used to prove Theorem \ref{Theorem:ExpDecay}. Section \ref{Sec:WWAtom} is devoted to applying Corollary \ref{Corallary_SimWeakContLimit} to a model describing spontaneous emission of atoms confined to large cavities in the rotating-wave approximation.
\section{Regularity of the Self-Energy Close to the Real-Axis}
\label{Sec:SelfEnergyReg}
Let $(H_L(\mu))_{L\in\mathbb{N}}$ be a sequence of Friedrichs Hamiltonians associated to triples $(E_0,K_L,\mu g_L)$. A function of great use in the study of Friedrichs Hamiltonians is the so-called \textit{self-energy} $\Sigma_L(z)$, defined by
$$
\Sigma_L(z):=\langle g_L,(K_L-z)^{-1}g_L\rangle=\int_{\mathbb{R}}\frac{1}{\lambda-z}\dd\nu_L(\lambda),\quad z\in\mathbb{C}\setminus\mathbb{R}.
$$
For $\varepsilon>0$ and $E\in\mathbb{R}$, we also define the self-energy at fixed distance from the real-axis by $\Sigma_L(\varepsilon;E)=\Sigma_L(E+i\varepsilon)$.  The purpose of this section is to establish certain regularity results on the self-energies $\Sigma_L(\varepsilon;\cdot)$ for small values of $\varepsilon>0$. Although the self-energies $\Sigma_L(\varepsilon;\cdot)$ are smooth for all $\varepsilon>0$, they may become sharply peaked close to the pure-points of $\nu_L$, for small values of $\varepsilon>0$. This becomes most clear when considering the imaginary part of the self-energies
$$
\Im\Sigma_L(\varepsilon;E)=\int_{-\infty}^\infty\frac{\varepsilon}{(\lambda-E)^2+\varepsilon^2}\dd\nu_L(\lambda).
$$
The above equation expresses that $\Im\Sigma_L(\varepsilon;\cdot)$ is just a convolution of the spectral measure with a Cauchy kernel, which can be viewed as a smoothened version of the measure, where its features at scale $\varepsilon$ are smeared out. In particular, if $\lambda_0$ is a pure-point of the measure $\nu_L$, then $\Im\Sigma_L(\varepsilon;\lambda_0)\geq \nu_L(\{\lambda_0\}) \varepsilon^{-1}$. The ``finer" the measures $\nu_L$ become in a neighbourhood of $E_0$, the less mass any given pure-point of $\nu_L$ can carry, leading to a suppression of the peaks:
\begin{lemma}
	\label{Lemma1}
	Let $(H_L(\mu))_{L\in\mathbb{N}}$ be a sequence of Friedrichs Hamiltonians associated to triples $(E_0,K_L,\mu g_L)$ satisfying Assumption \ref{Assumption} for some $\alpha\in(0,1]$. Then there is a compact neighbourhood $I$ of $E_0$ such that for any $\gamma\in(0,\alpha)$, there is a constant $C(\gamma)>0$ such that
	\begin{equation}
		\label{eq:Lemma1}
	\abs{\Sigma_L(\varepsilon;E_1)-\Sigma_L(\varepsilon;E_2)}\leq C(\gamma)\left(\abs{E_1-E_2}^\gamma + \frac{d(\nu_L,\nu_\infty)}{\varepsilon}\right),
	\end{equation}
for all $0<\varepsilon<1$, $L\in\mathbb{N}$ and $E_1,E_2\in I$.
\end{lemma}
\noindent
\textit{Proof. }By Assumption \ref{Assumption}, there is a compact interval $I_0=[E_0-\delta_0,E_0+\delta_0]$, $\delta_0>0$, on which the limiting measure $\nu_\infty$ has a bounded and Hölder continuous density $\rho_\infty$.
Without loss of generality, we assume that $|I_0|\leq 1$ and that the estimate 
\begin{equation}
	\label{eq:LevyMetricNeighbourhood}
	\sup_{\lambda\in I_0}\abs{F_{\nu_L}(\lambda)-F_{\nu_\infty}(\lambda)}\leq C_1d(\nu_L,\nu_\infty)
\end{equation}
 holds for some constant $C_1>0$. It follows from $|I_0|\leq 1$, that $|E_1-E_2|\leq |E_1-E_2|^\gamma$ for all $0\leq\gamma\leq1$ and $E_1,E_2\in I_0$, and we shall repeatedly make use of this fact. Let $I=[E_0-\delta,E_0+\delta]$, with $0<\delta<\delta_0$. Define $\Sigma_\infty(\varepsilon;E)=\int_{-\infty}^\infty (\lambda-E-i\varepsilon)^{-1}\dd \nu_\infty(\lambda)$ for $E\in\mathbb{R}$ and $\varepsilon>0$. We will begin by showing that for any $\gamma\in(0,\alpha)$ there is a constant $C(\gamma)>0$ such that for all $E_1,E_2\in I$
\begin{align}
	\label{eq1}
	\abs{\Sigma_\infty(\varepsilon;E_1)-\Sigma_\infty(\varepsilon;E_2)}\leq C(\gamma)|E_1-E_2|^\gamma,
\end{align}
for all $0<\varepsilon<1$. Then we will show that there is a $C_2>0$, such that for any $E\in I$,
\begin{equation}
	\label{eq:SelfEnergyLtoInftyEstimate}
	\abs{\Sigma_L(\varepsilon;E)-\Sigma_\infty(\varepsilon;E)}\leq C_2 \frac{d(\nu_L,\nu_\infty)}{\varepsilon},
\end{equation}
for all $0<\varepsilon<1$ and $L\in\mathbb{N}$, which then concludes the proof of the lemma via the triangle inequality. In order to prove that (\ref{eq1}) holds, we put $\eta=\delta_0-\delta>0$, and decompose, for any $E\in I$,
$$
\Sigma_\infty(\varepsilon;E)=\underset{:=Q_1(\varepsilon;E)}{\underbrace{\int_{\mathbb{R}\setminus I_0}\frac{1}{\lambda-E-i\varepsilon}\dd \nu_\infty(\lambda)}}+\underset{:=Q_2(\varepsilon;E)}{\underbrace{\int_{I_0\setminus B_\eta(E)}\frac{1}{\lambda-E-i\varepsilon}\dd \nu_\infty(\lambda)}}+\underset{:=Q_3(\varepsilon;E)}{\underbrace{\int_{B_\eta(E)}\frac{1}{\lambda-E-i\varepsilon}\dd \nu_\infty(\lambda)}} ,
$$
and prove that every summand obeys an estimate of the form (\ref{eq1}) individually. To that aim, let $E_1,E_2\in I$ and $\varepsilon\in(0,1)$ be arbitrary. For the first term, we simply observe that
$$
\abs{Q_1(\varepsilon;E_1)-Q_1(\varepsilon;E_2)}\leq\int_{\mathbb{R}\setminus I_0}\abs{\frac{1}{\lambda-E_1-i\varepsilon}-\frac{1}{\lambda-E_2-i\varepsilon}}\dd \nu_\infty(\lambda)\leq \frac{\abs{E_1-E_2}}{\eta^2}\leq\frac{\abs{E_1-E_2}^\gamma}{\eta^2}.
$$
 For the second term, we begin with the estimate
$$
\abs{Q_2(\varepsilon;E_1)-Q_2(\varepsilon;E_2)}\leq\int_{I_0}\abs{\frac{\mathbbm{1}_{B_\eta(E_1)^c}(\lambda)}{\lambda-E_1-i\varepsilon}-\frac{\mathbbm{1}_{B_\eta(E_2)^c}(\lambda)}{\lambda-E_2-i\varepsilon}}\rho_\infty(\lambda)\dd\lambda.
$$
In order to proceed, define
$$
A_{ij}=B_\eta(E_1)^i\cap B_\eta(E_2)^j, \quad i,j\in\{1,2\},
$$
with the understanding that $U^1=U$ and $U^2=U^c$ for any $U\subset\mathbb{R}$.  Then, by writing $I_0=\bigcup_{i,j=1}^2 \left(I_0\cap A_{ij}\right)$, we have that
$$
\abs{Q_2(\varepsilon;E_1)-Q_2(\varepsilon;E_2)}\leq \sum_{i,j=1}^2\int_{I_0 \cap A_{ij}}\abs{\frac{\mathbbm{1}_{B_\eta(E_1)^c}(\lambda)}{\lambda-E_1-i\varepsilon}-\frac{\mathbbm{1}_{B_\eta(E_2)^c}(\lambda)}{\lambda-E_2-i\varepsilon}}\rho_\infty(\lambda)\dd\lambda.
$$
For $(i,j)=(1,1)$, the integrand on the right hand side of the above equation vanishes, so that this combination does not contribute. For the term stemming from the combination $(i,j)=(1,2)$, the integrand simplifies to $|\lambda-E_2-i\varepsilon|^{-1}\rho_\infty(\lambda)$, which can be bounded from above in absolute value by $\eta^{-1}\sup_{\lambda\in I_0}|\rho_\infty(\lambda)|$ on $A_{12}\cap I_0\subset B_{\eta}(E_2)^c$. Moreover, one can easily verify that $|A_{12}|\leq |E_1-E_2|\leq |E_1-E_2|^\gamma$, so that this term can be bounded from above by $\sup_{\lambda \in I_0}\abs{\rho_\infty(\lambda)}\eta^{-1}|E_1-E_2|^\gamma$. The term stemming from the combination $(i,j)=(2,1)$ can be bounded analogously. Lastly, for the term corresponding to the combination $(i,j)=(2,2)$, the integrand becomes $|E_1-E_2| |\lambda-E_1-i\varepsilon|^{-1} |\lambda-E_2-i\varepsilon|^{-1}\rho_\infty(\lambda)$, which is bounded from above by $|E_1-E_2|^\gamma \eta^{-2}\rho_\infty(\lambda)$ on $I_0\cap A_{22}\subset B_\eta(E_1)^c\cap B_\eta(E_2)^c$. Using that $\nu_\infty$ is a probability measure, so that $\int_{I_0}\rho_\infty(\lambda)\dd \lambda \leq 1$, then shows that this term is bounded by $|E_1-E_2|^\gamma \eta^{-2}$. Putting everything together gives that
$$
\abs{Q_2(\varepsilon;E_1)-Q_2(\varepsilon;E_2)}\leq\left( 2\eta^{-1}\sup_{\lambda\in I_0}\abs{\rho_\infty(\lambda)}+\eta^{-2}\right)|E_1-E_2|^\gamma,
$$
so that this term also satisfies an estimate of the form (\ref{eq1}).
For the third summand, a change of variables yields that
\begin{align*}
	Q_3(\varepsilon;E_1)-Q_3(\varepsilon;E_2)&=\int_{B_\eta(E_1)}\frac{1}{\lambda-E_1-i\varepsilon}\dd \nu_\infty(\lambda)-\int_{B_\eta(E_2)}\frac{1}{\lambda-E_2-i\varepsilon}\dd \nu_\infty(\lambda)\\
	&=\int_{B_\eta(0)}\frac{\rho_\infty(E_1+\lambda)-\rho_\infty(E_2+\lambda)}{\lambda-i\varepsilon}\dd\lambda.
\end{align*}
In order to proceed, we write $(\lambda-i\varepsilon)^{-1}=\lambda/(\lambda^2+\varepsilon^2) + i \varepsilon/(\lambda^2+\varepsilon^2)$ and estimate the terms stemming from the real and imaginary part separately. For the term coming from the imaginary part, we use that $\int_{\mathbb{R}}\varepsilon(\lambda^2+\varepsilon^2)\dd\lambda =\pi$, to find
$$
\abs{\Im Q_3(\varepsilon;E_1)-\Im Q_3(\varepsilon;E_2)}= \abs{\int_{B_\eta(0)}\frac{\varepsilon(\rho_\infty(E_1+\lambda)-\rho_\infty(E_2+\lambda))}{\lambda^2+\varepsilon^2}\dd\lambda}\leq \pi M\abs{E_1-E_2}^\alpha,
$$
where $M$ is the Hölder constant for $\rho_\infty$ on $I_0$. For the term stemming from the real part, we use the fact that the function $\lambda\mapsto \lambda/(\lambda^2+\varepsilon^2)$ is odd and that the region integrated over is even. This gives the estimate
\begin{align*}
	\abs{\Re Q_3(\varepsilon;E_1)-\Re Q_3(\varepsilon;E_2)}&=\abs{\int_{B_\eta(0)}\frac{\lambda(\rho_\infty(E_1+\lambda)-\rho_\infty(E_2+\lambda))}{\lambda^2+\varepsilon^2}\dd\lambda}\\
	&\leq 2M\int_{B_\eta(0)} \frac{\min\{\abs{E_1-E_2}^\alpha,\abs{\lambda}^\alpha\}}{\abs{\lambda}}\dd \lambda\\
	&\leq 4M\left[\frac{1}{\alpha} +\abs{ \ln(\frac{\abs{E_1-E_2}}{\eta})}\right]\abs{E_1-E_2}^\alpha.
\end{align*}
Since $\sup_{\lambda\in[0,1]}\lambda^{\alpha-\gamma}|\ln(\lambda)|<\infty$ and $|E_1-E_2|^\alpha\leq |E_1-E_2|^\gamma$ , this shows that
$$
\abs{Q_3(\varepsilon;E_1)-Q_3(\varepsilon;E_2)}\leq \left(\pi M +\frac{4M}{\alpha} +|\ln(\eta)|+\sup_{\lambda\in[0,1]}\lambda^{\alpha-\gamma}|\ln(\lambda)|\right)|E_1-E_2|^\gamma,
$$
which concludes the proof of (\ref{eq1}). We now proceed with proving (\ref{eq:SelfEnergyLtoInftyEstimate}). To that aim fix $0<\varepsilon<1$ and $L\in\mathbb{N}$, and let $\sigma>0$ be such that 
\begin{equation*}
	F_{\nu_\infty}(\lambda-\sigma)-\sigma\leq F_{\nu_L}(\lambda)\leq F_{\nu_\infty}(\lambda+\sigma)+\sigma,
\end{equation*}
for all $\lambda\in\mathbb{R}$. Then, by monotonicity of $F_{\nu_\infty}$,
\begin{equation}
	\label{eq:UniformEstimateLevy}
	\abs{F_{\nu_L}(\lambda)-F_{\nu_\infty}(\lambda)}\leq \sigma + \int_{[\lambda-\sigma,\lambda+\sigma]}\dd\nu_\infty(\kappa),
\end{equation}
for all $\lambda\in\mathbb{R}$. By using integration by parts and by splitting the integration region, we have that for any $E\in I$,
\begin{equation*}
	\abs{\Sigma_L(\varepsilon;E)-\Sigma_\infty(\varepsilon;E)}\leq \int_{B_\eta(E)}\frac{\abs{F_{\nu_L}(\lambda)-F_{\nu_\infty}(\lambda)}}{(\lambda-E)^2+\varepsilon ^2}\dd\lambda+\int_{B_\eta(E)^c}\frac{\abs{F_{\nu_L}(\lambda)-F_{\nu_\infty}(\lambda)}}{(\lambda-E)^2+\varepsilon ^2}\dd\lambda.
\end{equation*}
For the integral over $B_\eta(E)\subset I_0$, we use the estimate (\ref{eq:LevyMetricNeighbourhood}) together with the fact that $\int_{-\infty}^\infty(\lambda^2+\varepsilon^2)^{-1}\dd\lambda =\pi\varepsilon^{-1}$, to bound it from above by $C_1\pi d(\nu_L,\nu_\infty)\varepsilon^{-1}$. For the integral over $B_{\eta}(E)^c$, we use the estimate (\ref{eq:UniformEstimateLevy}) together with Fubini--Tonelli to obtain
\begin{align*}
	\int_{B_\eta(E)^c}\frac{\abs{F_{\nu_L}(\lambda)-F_{\nu_\infty}(\lambda)}}{(\lambda-E)^2+\varepsilon ^2}\dd\lambda\leq \frac{\pi\sigma}{\varepsilon} + \int_{-\infty}^\infty\left(\int_{-\infty}^\infty \frac{f(\lambda,\kappa)}{(\lambda-E)^2+\varepsilon^2}\dd \lambda\right)\dd\nu_{\infty}(\kappa),
\end{align*}
where 
\begin{equation*}
	f(\lambda,\kappa)=
	\begin{cases}
		1 & \text{if }\lambda\in \overline{B_{\sigma}(\kappa)}\cap B_{\eta}(E)^c\\
		0 & \text{else.}
	\end{cases}
\end{equation*}
Here we also made use of the fact that $\mathbbm{1}_{\overline{B_\sigma(\kappa)}}(\lambda)=\mathbbm{1}_{\overline{B_\sigma(\lambda)}}(\kappa)$ for all $\lambda,\kappa\in\mathbb{R}$. Thus
\begin{equation*}
	\int_{-\infty}^\infty\left(\int_{-\infty}^\infty \frac{f(\lambda,\kappa)}{(\lambda-E)^2+\varepsilon^2}\dd \lambda\right)\dd\nu_{\infty}(\kappa)\leq \frac{2\sigma}{\eta^2}\leq \frac{2\sigma}{\eta^2\varepsilon},
\end{equation*}
so that, by taking the limit $\sigma\to d(\nu_L,\nu_\infty)$,
$$
\abs{\Sigma_L(\varepsilon;E)-\Sigma_\infty(\varepsilon;E)}\leq \left(C_1\pi +\pi+ \frac{2}{\eta^2}\right)\frac{d(\nu_L,\nu_\infty)}{\varepsilon},
$$
which concludes the proof of (\ref{eq:SelfEnergyLtoInftyEstimate}), and thus also the proof of the lemma.
\hspace*{\fill}$\square$\\

\noindent
Next, we will show that for suitable values of $\varepsilon>0$ and $L\in\mathbb{N}$, the self-energy $\Sigma_L(E_0;\varepsilon)$ is well approximated by $\Delta(E_0)+i\Gamma(E_0)$, where the energy shift $\Delta(E_0)$ and the decay rate $\Gamma(E_0)$ are defined in Equations (\ref{eq:EnergyShift}) and (\ref{eq:DecayRate}). 
\begin{lemma}
	\label{Lemma3}
	Let $(H_L(\mu))_{L\in\mathbb{N}}$ be a sequence of Friedrichs Hamiltonians associated to triples $(E_0,K_L,\mu g_L)$ satisfying Assumption \ref{Assumption} for some $\alpha\in(0,1]$. Then there is a compact neighbourhood $I$ of $E_0$, such that for all $E\in I$, the quantity
	\begin{align*}
		\Delta(E)&=\lim_{\delta\to 0^+}\int_{\mathbb{R}\setminus B_\delta(E)}\frac{1}{\lambda - E}\dd \nu_\infty(\lambda)
	\end{align*}
	exists. Moreover, there is a $C>0$ such that for all  $0<\varepsilon<1$ and $L\in\mathbb{N}$ we have that
	\begin{align}
		\label{eq:SelfEnergyApprox}
		\abs{\Sigma_L(\varepsilon;E_0)-\left(\Delta(E_0) + i \pi\rho_\infty(E_0)\right)}&\leq C\left[\frac{d(\nu_L,\nu_\infty)}{\varepsilon} + \varepsilon^{\alpha}\right].
	\end{align}
\end{lemma}
\noindent
\textit{Proof. }By Assumption \ref{Assumption}, there is a compact interval $I$ around $E_0$, on which $\nu_\infty$ has a bounded and Hölder continuous density $\rho_\infty$. The existence of $\Delta(E)$ for any $E\in I$ follows easily from the Hölder continuity of $\rho_\infty$. In order to prove the estimate (\ref{eq:SelfEnergyApprox}), it is sufficient to prove that for all $0<\varepsilon<1$,
\begin{align}
	\abs{\Im\Sigma_\infty(\varepsilon;E_0)-\pi\rho_\infty(E_0)}&\leq C \varepsilon^{\alpha},\\
	\abs{\Re\Sigma_\infty(\varepsilon;E_0)-\Delta(E_0)}&\leq C \varepsilon^{\alpha},
\end{align}
for some $C>0$, using (\ref{eq:SelfEnergyLtoInftyEstimate}). We begin by proving the estimate for the imaginary part. To that aim, we note that
$$
\Im\Sigma_\infty(\varepsilon;E_0)-\pi\rho_\infty(E_0)=\int_{-\infty}^\infty\frac{\varepsilon}{(\lambda-E_0)^2+\varepsilon^2}\dd\nu_\infty(\lambda) - \int_{-\infty}^\infty\frac{\varepsilon}{(\lambda-E_0)^2+\varepsilon^2}\rho_\infty(E_0)\dd\lambda.
$$
Next, we split the integration regions into $B_\delta(E_0)\cup B_\delta(E_0)^c$, where $0<\delta<1$ is chosen such that $B_\delta(E_0)\subset I$. It then follows from the estimates
\begin{align*}
	\abs{\int_{B_\delta(E_0)^c}\frac{\varepsilon}{(\lambda-E_0)^2+\varepsilon^2}\dd\nu_\infty(\lambda)}&\leq \frac{\varepsilon}{\delta^2},\\
	\abs{\int_{B_\delta(E_0)^c}\frac{\varepsilon}{(\lambda-E_0)^2+\varepsilon^2}\rho_\infty(E_0)\dd\lambda}&\leq 2\sup_{\lambda\in I}\abs{\rho_\infty(\lambda)}\int_\delta^\infty\frac{\varepsilon}{\lambda^2+\varepsilon^2}\dd\lambda\leq 2\sup_{\lambda\in I}\abs{\rho_\infty(\lambda)}\frac{\varepsilon}{\delta},
\end{align*}
that
\begin{align*}
	\abs{\Im\Sigma_\infty(\varepsilon;E_0)-\pi\rho_\infty(E_0)}&\leq \int_{B_\delta(E_0)} \frac{\varepsilon\abs{\rho_\infty(\lambda)-\rho_\infty(E_0)}}{(\lambda-E_0)^2+\varepsilon^2}\dd\lambda + \frac{\varepsilon}{\delta^2}+ 2\sup_{\lambda\in I}\abs{\rho_\infty(\lambda)}\frac{\varepsilon}{\delta}\\
	&\leq 2M\int_0^\delta \frac{\varepsilon\lambda^{\alpha}}{\lambda^2+\varepsilon^2} \dd\lambda + \frac{\varepsilon}{\delta^2}+ 2\sup_{\lambda\in I}\abs{\rho_\infty(\lambda)}\frac{\varepsilon}{\delta}\\
	&\leq 2M\varepsilon^{\alpha}\int_0^{\delta/\varepsilon}\frac{\lambda^\alpha}{\lambda^2+1}\dd\lambda +(1+2\sup_{\lambda\in I}\abs{\rho_\infty(\lambda)})\frac{\varepsilon}{\delta^2},
\end{align*}
where $M$ is the Hölder constant of $\rho_\infty$ on $I$. Since $\int_0^\infty \lambda^\alpha/(\lambda^2+1)\dd\lambda<\infty$, the estimate for the imaginary part now follows since $0<\varepsilon<1$, so that $\varepsilon<\varepsilon^\alpha$. For the estimate involving the real part, we write
$$
\Re\Sigma_\infty(\varepsilon;E_0)-\Delta(E_0)=\lim_{\sigma\to 0^+}\int_{-\infty}^\infty \left[\frac{\lambda-E_0}{(\lambda-E_0)^2+\varepsilon^2}-\frac{\mathbbm{1}_{B_\sigma(E_0)^c}(\lambda)}{\lambda-E_0}\right]\dd\nu_\infty(\lambda).
$$
Again, we split the integration region into $B_\delta(E_0)\cup B_\delta(E_0)^c$ for sufficiently small $0<\delta<1$, such that $B_\delta(E_0)\subset I$. Then, for any $\sigma<\delta$, the integral over $B_\delta(E_0)$ can be estimated as
\begin{align*}
&\abs{\int_{B_\delta(E_0)} \left[\frac{\lambda-E_0}{(\lambda-E_0)^2+\varepsilon^2}-\frac{\mathbbm{1}_{B_\sigma(E_0)^c}(\lambda)}{\lambda-E_0}\right]\dd\nu_\infty(\lambda)}\\
&\quad\leq \int_{B_\delta(E_0)}\abs{ \left[\frac{\lambda-E_0}{(\lambda-E_0)^2+\varepsilon^2}-\frac{\mathbbm{1}_{B_\sigma(E_0)^c}(\lambda)}{\lambda-E_0}\right](\rho_\infty(\lambda)-\rho_\infty(E_0))}\dd\lambda\\
&\quad\leq M\int_{B_\delta(E_0)}\frac{|\lambda-E_0|^{\alpha-1}(\varepsilon^2 + (\lambda-E_0)^2\mathbbm{1}_{B_\sigma(E_0)}(\lambda))}{(\lambda-E_0)^2+\varepsilon^2}\dd\lambda\\
&\quad\leq 2M\varepsilon^{\alpha}\int_0^{\delta/\varepsilon}\frac{\lambda^{\alpha-1}}{\lambda^2+1}\dd\lambda + 2M\frac{\sigma^{2+\alpha}}{\varepsilon^2}.
\end{align*}
For the integral over $B_\delta(E_0)^c$, we note that $\mathbbm{1}_{B_\sigma(E_0)^c}(\lambda)=1$ for $\lambda\in B_\delta(E_0)^c$. Hence
\begin{align*}
	&\abs{\int_{B_\delta(E_0)^c} \left[\frac{\lambda-E_0}{(\lambda-E_0)^2+\varepsilon^2}-\frac{\mathbbm{1}_{B_\sigma(E_0)^c}(\lambda)}{\lambda-E_0}\right]\dd\nu_\infty(\lambda)}\leq\frac{\varepsilon^2}{\delta^3} \int_{B_\delta(E_0)^c}\dd\nu_\infty(\lambda)\leq \frac{\varepsilon^2}{\delta^3}.
\end{align*}
By taking $\sigma\to 0^+$ and using that $K:=\int_0^\infty\lambda^{\alpha-1}(\lambda^2+1)^{-1}\dd\lambda<\infty$, we therefore obtain
$$
\abs{\Re\Sigma_\infty(\varepsilon;E_0)-\Delta(E_0)}\leq\left(2KM + \delta^{-3}\right)\varepsilon^\alpha,
$$
which concludes the proof of the lemma.
\hspace*{\fill}$\square$\\

\section{Exponential Decay}
\label{Sec:ExpDecay}
In this section we shall make use of the regularity results on the self-energies provided in the previous section in order to prove that the excited state $\chi$ decays in an approximately exponential manner under favourable conditions, following the outline discussed in the introduction. We begin by deriving an explicit formula for the density of the smeared spectral measure in terms of the self-energy.
\begin{lemma}
	\label{Lemma3.1}
	Let $H(\mu)$ be a Friedrichs Hamiltonian associated to the triple $(E_0,K,\mu g)$ with self-energy $\Sigma(E+i\varepsilon)\equiv \Sigma(\varepsilon; E)=\langle g, (K-E-i\varepsilon)^{-1}g\rangle$. For $\varepsilon>0$, let
	\begin{equation}
		p_\varepsilon(\lambda)=\frac{1}{\pi}\frac{\varepsilon}{\lambda^2+\varepsilon^2},\quad \lambda\in\mathbb{R}.
	\end{equation}
Let $\sigma^\mu$ denote the spectral measure associated to $\chi$ by $H(\mu)$. Then
\begin{equation}
	\label{eq:SmearedDensity}
	\rho^\mu(\varepsilon;E):=(p_\varepsilon * \sigma^\mu)(E)=\frac{1}{\pi}\frac{\varepsilon+\mu^2 \Im\Sigma(\varepsilon;E)}{(E-E_0+\mu^2\Re\Sigma(\varepsilon;E))^2 + (\varepsilon+\mu^2 \Im\Sigma(\varepsilon;E))^2}.
\end{equation}
\end{lemma}
\noindent
\textit{Proof. }It follows from functional calculus that
\begin{equation*}
	(p_\varepsilon*\sigma^\mu)(E)=\int_{-\infty}^\infty\frac{1}{\pi}\frac{\varepsilon}{(E-\lambda)^2+\varepsilon^2}\dd\sigma^\mu(\lambda)=\frac{1}{\pi}\Im\left\langle\chi, (H(\mu)-E-i\varepsilon)^{-1}\chi\right\rangle.
\end{equation*}
In order to conclude the proof, we must now analyse the structure of the expectation value of the resolvent $\langle \chi,(H(\mu)-E-i\varepsilon)^{-1}\chi\rangle$. To that aim, let $z\in\mathbb{C}$ with $\Im(z)>\mu$ be arbitrary. Then, the Neumann series
$$
(H(\mu)-z)^{-1}=(H(\mu)-z)^{-1}\sum_{n=0}^{\infty}\left[(H(0) - H(\mu))(H(0)-z)^{-1}\right]^n,
$$
is absolutely convergent, since $W(\mu):=H(0)-H(\mu)$ can be extended to a bounded operator with norm $\mu$. Let $P_{\mathbb{C}}$ and $P_{\mathcal{K}}$ be the orthogonal projections onto $\mathbb{C}$ and $\mathcal{K}$ respectively. Then it is easy to verify that
\begin{enumerate}[(i)]
	\item $P_\mathbb{C}=|\chi\rangle\langle\chi|$,
	\item $P_\mathbb{C}W(\mu)= W(\mu) P_{\mathcal{K}}$,
	\item $P_{\mathcal{K}}W(\mu)=W(\mu) P_{\mathbb{C}}$.
\end{enumerate}
Making use of these relations, we find that
\begin{align*}
	\langle\chi,(H(\mu)-z)^{-1}\chi\rangle&=(E_0-z)^{-1}\sum_{n=0}^\infty \left[\mu^2\langle g,(K-z)^{-1}g\rangle (E_0-z)^{-1}\right]^n=\frac{1}{E_0-z-\mu^2\Sigma(z)}.
\end{align*}
By analyticity, this equality extends to all of $\{z\in\mathbb{C}\text{ : }\Im(z)>0\}$. The lemma now follows by taking the imaginary part of the equation above at $z=E+i\varepsilon$.
\hspace*{\fill}$\square$\\

\noindent
We are now ready to prove the main result of this paper, namely the approximately exponential decay of the survival amplitude. The idea is to first appropriately ``smear out" the spectral measure associated to $\chi$ by $H_L(\mu)$ and then to show that the resulting density is well approximated in $L^1(\mathbb{R})$ by a Cauchy distribution using the regularity results proven in Section \ref{Sec:SelfEnergyReg}, similar to the proof of \cite[Theorem 4.1]{JensenNenciuOddDim}.\\

\noindent
\textit{Proof of Theorem \ref{Theorem:ExpDecay}. }Let $(H_L(\mu))_{L\in\mathbb{N}}$ be a sequence of Friedrichs Hamiltonians associated to triples $(E_0,K_L,\mu g_L)$ satisfying Assumption \ref{Assumption} for some $\alpha\in(0,1]$. Let $\Delta(E_0)$ and $\Gamma(E_0)$ be defined by (\ref{eq:EnergyShift}) and (\ref{eq:DecayRate}) and put
$$
R_L^\mu(t):=\langle\chi,\exp(-itH_L(\mu))\chi\rangle - e^{-i(E_0-\mu^2\Delta(E_0))t}e^{-\mu^2\Gamma(E_0)t},\quad t\geq 0.
$$
 Fix a $\gamma\in (0,\alpha)$. Throughout this proof, we shall abbreviate $\varepsilon\equiv \varepsilon(\mu,L):=d(\nu_L,\nu_\infty)\mu^{-2\gamma}$. Recall that we aim to prove the existence of a $C(\gamma)>0$ such that, for all $t\geq 0$,
  \begin{equation}
 	\abs{R^\mu_L(t)}\leq C(\gamma)\left(\mu^{2\gamma} + \left(\frac{d(\nu_L,\nu_\infty)}{\mu^{2}}\right)^{\alpha} \right)\exp(\frac{d(\nu_L,\nu_\infty)}{\mu^{2\gamma}}t),
 \end{equation}
for all $\mu\geq 0$ and $L\in\mathbb{N}$. Letting $\rho_L^\mu(\varepsilon;\cdot)$ and $\tilde{\rho}_L^\mu(\varepsilon;\cdot)$ be defined by Equations (\ref{eq:rho}) and (\ref{eq:rhotilde}) respectively, the above estimate follows if we can prove that there is a $C(\gamma)>0$, such that $$\norm{\rho_L^\mu(\varepsilon;\cdot)-\tilde{\rho}_L^\mu(\varepsilon;\cdot)}_{L^1(\mathbb{R})}\leq C(\gamma)\left(\mu^{2\gamma} + \left(\frac{d(\nu_L,\nu_\infty)}{\mu^2}\right)^\alpha\right),$$
for all $\mu\geq 0$ and $L\in\mathbb{N}$, using the estimate (\ref{eq:3}).
Let $\mu_0,q_0,\varepsilon_0>0$ be arbitrary. Since $\norm{\rho_L^\mu(\varepsilon;\cdot)-\tilde{\rho}_L^\mu(\varepsilon;\cdot)}_{L^1(\mathbb{R})}\leq 2$, we may, by choosing $C(\gamma)>0$ sufficiently large, restrict to considering the case where
\begin{equation}
	\label{eq:restrictions}
	\mu<\mu_0 \text{ and }\frac{d(\nu_L,\nu_\infty)}{\mu^2}<q_0,
\end{equation}
and, consequently, also to the case where
\begin{equation}
	\label{eq:epsilon}
	\varepsilon \equiv \varepsilon(\mu,L):= \frac{d(\nu_L,\nu_\infty)}{\mu^{2\gamma}}<\varepsilon_0,
\end{equation}
by taking $\mu_0<1$. We also choose $\varepsilon_0,q_0<1$. By Lemma \ref{Lemma1} there is a compact interval $I$ around $E_0$ and a constant $C_1(\gamma)>0$ such that the estimate (\ref{eq:Lemma1}) holds for all $E_1,E_2\in I$. Also fix a $\delta\in(0,2)$, which we will use to control the size of the interval around $E_0$ on which we will perform our estimates. We choose $\mu_0<1$ such that for any $\mu<\mu_0$ the closed interval $J(\mu)$ centred about $E_0-\mu^2\Delta(E_0)$ with $|J(\mu)|=\mu^{2-\delta}$ satisfies $J(\mu)\subset I$.
Let us now fix arbitrary $\mu\geq 0$ and $L\in\mathbb{N}$ such that (\ref{eq:restrictions}) and (\ref{eq:epsilon}) hold. Put $J\equiv J(\mu)$. 
Define
\begin{align*}
	G(E)&=E_0-E-i\varepsilon -\mu^2 \Sigma_L(\varepsilon;E)\\
	\tilde{G}(E)&=E_0-E-i\varepsilon -\mu^2 (\Delta(E_0)+i\Gamma(E_0)),
\end{align*}
so that $\rho(E)\equiv \rho_L^\mu(\varepsilon;E)=\pi^{-1}\Im G(E)^{-1}$ by Lemma \ref{Lemma3.1} and $\tilde{\rho}(E)\equiv \tilde{\rho}_L^\mu(\varepsilon;E)=\pi^{-1}\Im \tilde{G}(E)^{-1}$. The estimates in Lemma \ref{Lemma1} and Lemma \ref{Lemma3} imply that there is a $C_2(\gamma)>0$ such that for all $E\in J$,
\begin{equation*}
	\abs{G(E)-\tilde{G}(E)}\leq C_2(\gamma)\mu^2\left[ \varepsilon^\alpha +\mu^{2\gamma}+ \abs{E-E_0}^\gamma\right].
\end{equation*}
Since $\abs{\tilde{G}(E)}\geq \mu^2\Gamma(E_0)$, we may, by appropriate choice of $\mu_0$ and $\varepsilon_0$, assume that $\abs{G(E)}\geq 2^{-1}\abs{\tilde{G}(E)}$ for all $E\in J$.
 It follows that
\begin{align*}
	\norm{(\rho-\tilde{\rho})\mathbbm{1}_J}_{L^1(\mathbb{R})}&\leq  \frac{1}{\pi}\int_J\abs{\frac{1}{G(E)}-\frac{1}{\tilde{G}(E)}}\dd E\leq \frac{2}{\pi}\int_J\frac{\abs{G(E)-\tilde{G}(E)}}{\abs{\tilde{G}(E)}^2}\dd E\\
	&\leq\frac{4C_2(\gamma)}{\pi\Gamma(E_0)}\int_{0}^{\mu^{2-\delta}}\frac{\mu^2\Gamma(E_0)}{E^2 + (\mu^2\Gamma(E_0))^2}\left[\varepsilon^\alpha + \mu^{2\gamma} + E^\gamma +\mu^{2\gamma}\abs{\Delta(E_0)}^\gamma\right]\dd E\\
	&\leq C_3(\gamma)\left[\varepsilon^\alpha + \mu^{2\gamma} +\mu^{2\gamma} \int_0^{\mu^{-\delta}/\Gamma(E_0)}\frac{E^\gamma}{E^2+1}\dd E\right],
\end{align*}
where in the second-to-last line, we used that $\abs{E-E_0}^\gamma\leq \abs{E-E_0+\mu^2\Delta(E_0)}^\gamma + (\mu^2\abs{\Delta(E_0)})^\gamma$. By virtue of the fact that $\int_0^\infty x^\gamma/(x^2+1)\dd x<\infty$, it follows that there is a $C_4(\gamma)>0$ such that $\norm{(\rho-\tilde{\rho})\mathbbm{1}_J}_{L^1(\mathbb{R})}\leq C_4(\gamma)(\varepsilon^\alpha + \mu^{2\gamma})$. Now, the density $\tilde{\rho}$ is concentrated on $J$, namely
\begin{align*}
	\norm{\tilde{\rho}\mathbbm{1}_{\mathbb{R}\setminus J}}_{L^1(\mathbb{R})}&=1- \pi^{-1}\int_{-\abs{J}/2}^{\abs{J}/2} \frac{\varepsilon + \mu^2 \Gamma(E_0)}{E^2+(\varepsilon +\mu^2\Gamma(E_0))^2}\dd E\leq\frac{2}{\pi}\int_{\abs{J}/2}^\infty \frac{\varepsilon+\mu^2 \Gamma(E_0)}{E^2}\dd E\\
	&\leq \frac{4}{\pi}\frac{\varepsilon + \mu^2\Gamma(E_0)}{|J|}\leq  C_5(\gamma)\left(\frac{d(\nu_L,\nu_\infty)}{\mu^{2\gamma}\mu^{2-\delta}} + \mu^{\delta}\right).
\end{align*}
This suggests that we choose $\delta=2\gamma$. Since $\rho$ and $\tilde{\rho}$ are both non-negative and have unit norm in $L^1(\mathbb{R})$, it follows that
\begin{align*}
	\int_{\mathbb{R}\setminus J}\rho(E)\dd E&=1-\int_{J}\rho(E)\dd E=1-\int_{J}(\rho(E)-\tilde{\rho}(E))\dd E - \int_{J}\tilde{\rho}(E)\dd E\\
	&=\int_{\mathbb{R}\setminus J}\tilde{\rho}(E)\dd E - \int_{J}(\rho(E)-\tilde{\rho}(E))\dd E.
\end{align*}
In particular, this implies the estimate $\norm{\rho \mathbbm{1}_{\mathbb{R}\setminus J}}_{L^1(\mathbb{R})} \leq \norm{\tilde{\rho}\mathbbm{1}_{\mathbb{R}\setminus J}}_{L^1(\mathbb{R})} + \norm{(\rho-\tilde{\rho})\mathbbm{1}_J}_{L^1(\mathbb{R})}$. Hence, using that $\varepsilon<d(\nu_L,\nu_\infty)/\mu^2<1$, we find that
$$
	\norm{\rho-\tilde{\rho}}_{L^1(\mathbb{R})}\leq C_6(\gamma)\left(\mu^{2\gamma} + \left(\frac{d(\nu_L,\nu_\infty)}{\mu^{2}}\right)^{\alpha} \right),
$$
which is the desired inequality.
\hspace*{\fill}$\square$
\section{Spontaneous Emission in Large Cavities}
\label{Sec:WWAtom}
	In this section we will show that, under certain regularity assumptions, the confined spin-boson Hamiltonian in the rotating-wave approximation $H_L^{\text{SB}}(E_0,\mu g_L)$, defined in \eqref{eq:Hamiltonian}, satisfies Assumption \ref{Assumption} when restricted to the one-excitation sector. Recall from the introduction that this restricted Hamiltonian is unitarily equivalent to the Friedrichs Hamiltonian associated to the triple $(E_0,K_L,\mu g_L)$, where $K_L$ is the self-adjoint operator acting by $(K_Lf)(k)=|k|f(k)$ on $\ell^2(\Lambda_L)$ and where $\Lambda_L=2\pi L^{-1}\mathbb{Z}^3$. We will derive the coupling functions $g_L$ from a function $g\in L^2(\mathbb{R}^3)$ satisfying the following regularity assumption:
	\begin{assumptionA}{R}
		\label{AssumptionR}
		The function $g\in \text{L}^2(\mathbb{R}^3)$ satisfies $g\in H^1(\mathbb{R}^3)\cap C^1(\mathbb{R}^3)$, $\norm{g}_{L^2(\mathbb{R}^3)}=1$ and $g(k)|\leq C(1+|k|^2)^{-1/2}$ for some $C>0$.
	\end{assumptionA}
	\noindent
 Given a function $g\in\text{L}^2(\mathbb{R}^3)$ satisfying Assumption \ref{AssumptionR} and an $L>0$, we define $g_L\in \ell^2(\Lambda_L)$ by $g_L(k)=C_L(2\pi/L)^{3/2}g(k)$, where the normalisation constant $C_L$ is given by
	\begin{equation}
		C_L=\left[\left(\frac{2\pi}{L}\right)^3\sum_{k\in\Lambda_L}\abs{g(k)}^2\right]^{-1/2}.
	\end{equation}
	\begin{lemma}
		\label{Lemm_Properties}
		Let $g\in L^2(\mathbb{R}^3)$ satisfy Assumption \ref{AssumptionR} and let $\nu_L$ denote the spectral measure associated to $g_L$ by $K_L$. Then $\nu_L$ converges weakly to an absolutely continuous measure $\nu_\infty$ as $L\to\infty$. The Radon-Nikodym derivative of $\nu_\infty$ is given by
		\begin{equation}
			\label{eq_DefRhoInfty}
			\rho_\infty(E)=\mathbbm{1}_{(0,\infty)}(E)\cdot\norm{g\restriction_{S(E)}}_{\text{L}^2(S(E),\dd\omega_E)}^2, \quad E\in\mathbb{R},
		\end{equation}
		where $S(E)=\{k\in\mathbb{R}^3\text{ : }\abs{k}=E\}$ denotes the energy-shell of energy $E$, equipped with the natural surface measure $\omega_E$. Moreover, the density $\rho_\infty$ is locally Hölder continuous of degree $\alpha$, for any $\alpha\in(0,1/2)$, and there is a $K>0$ such that
		\begin{equation}
			\label{eq:SBdensity}
					d(\nu_L,\nu_\infty)\leq \sup_{\lambda\in\mathbb{R}}\abs{F_{\nu_L}(\lambda)-F_{\nu_\infty}(\lambda)}\leq K L^{-1}.
		\end{equation}
		In particular, for any $E_0>0$ with $\rho_\infty(E_0)>0$, the sequence of Friedrichs Hamiltonians associated to the triples $(E_0,K_L,\mu g_L)$ satisfies Assumption \ref{Assumption} for any $\alpha\in(0,1/2)$.
	\end{lemma}
	\noindent
	\textit{Proof. }Since $g\in H^1(\mathbb{R}^3)$, the map
	$$
	E\mapsto g\restriction_{S(E)}(E\cdot )
	$$
	is Hölder continuous of order $\alpha$ as an $\text{L}^2(S(1),\dd \omega_1)$-valued function for any $\alpha\in(0,1/2)$, see for example \cite[Theorem IX.40]{ReedSimonVol2}. Now, $\norm{g\restriction_{S(E)}}_{\text{L}^2(S(E),\dd\omega_E)}^2=E^2\norm{g\restriction_{S(E)}(E\cdot)}_{\text{L}^2(S(1),\dd\omega_1)}^2$, which proves local Hölder continuity of $\rho_\infty$ defined in Equation (\ref{eq_DefRhoInfty}). Moreover
	\begin{equation*}
		\rho_\infty(E)=E^2\norm{g\restriction_{S(E)}}^2_{L^2(S(1),\dd\omega_1)}\leq 4\pi C \frac{E^2}{1+E^2},
	\end{equation*}
	showing that $\sup_{\lambda\geq 0}\rho_\infty(\lambda)<\infty$.
	Let $\lambda\in\mathbb{R}\cup\{+\infty\}$ be arbitrary. For any $L\in\mathbb{N}$, the cumulative distribution function associated to $\nu_L$ is given by 
	\begin{equation*}
		F_{\nu_L}(\lambda) =C_L^2\left(\frac{2\pi}{L}\right)^3\sum_{k\in\Lambda_L, \abs{k}\leq \lambda}\abs{g(k)}^2.
	\end{equation*}
	Letting $\nu_\infty$ denote the probability measure associated to the density $\rho_\infty$, a change of variables yields that
	\begin{equation*}
	F_{\nu_\infty}(\lambda)=\int_{(-\infty,\lambda]}\rho_\infty(E)\dd E=\int_{\abs{k}\leq \lambda}\abs{g(k)}^2\dd k.
	\end{equation*}
	For $k\in \mathbb{R}^3$, let $Q_L(k)=\bigtimes_{i=1}^3[k_i-\pi/L,k_i+\pi/L)$ be the cube of side-length $2\pi/L$ centred around $k$ and put $\tilde{B}_\lambda^L=\bigcup_{k\in\Lambda_L,\abs{k}\leq \lambda}Q_L(k)$. Then, letting $A\Delta B$ denote the symmetric difference of two sets $A,B\subset\mathbb{R}^3$ and abbreviating $f\equiv \abs{g}^2$,
	\begin{align*}
		&\abs{\left(\frac{2\pi}{L}\right)^3\sum_{k\in\Lambda_L,\abs{k}\leq\lambda}f(k)-\int_{\abs{k}\leq \lambda}f(k)\dd k}\\
		&\quad\leq \sum_{k\in\Lambda_L,\abs{k}\leq \lambda}\int_{Q_L(k)}\abs{f(k)-f(q)}\dd q + \int_{B_\lambda(0)\Delta \tilde{B}^L_\lambda}f(q)\dd q\\
		&\quad\leq \frac{K_1}{L}\sum_{k\in\Lambda_L,\abs{k}\leq \lambda}\int_{Q_L(k)}\abs{\nabla f(q)}\dd q +\int_{S_\lambda^L}f(q)\dd q\\
		&\quad\leq  \frac{K_1}{L}\int_{\mathbb{R}^3}\abs{\nabla f(q)}\dd q + \int_{\lambda - \sqrt{12}\pi /L}^{\lambda+\sqrt{12}\pi/L}\rho_\infty(E)\dd E\\
		&\quad\leq \frac{2K_1}{L}\norm{g}_{H^1(\mathbb{R}^3)}^2 + \sup_{\lambda\geq 0}\rho_\infty(\lambda)\frac{\sqrt{48}\pi}{L}=:K_2L^{-1},
	\end{align*}
	using Poincare's inequality and where $S_\lambda^L=B_{\lambda + \sqrt{12}\pi/L}(0)\setminus B_{\lambda - \sqrt{12}\pi/L}(0)\supset B_\lambda(0)\Delta \tilde{B}_\lambda^L$. In particular, using the estimate for $\lambda=+\infty$, we have that $\abs{C_L^2-1}\leq 2K_2L^{-1}$ for sufficiently large $L\in\mathbb{N}$. Hence
	$$
	\abs{F_{\nu_L}(\lambda)-F_{\nu_\infty}(\lambda)}=\abs{C_L^2\left(\frac{2\pi}{L}\right)^3\sum_{k\in\Lambda_L, \abs{k}\leq \lambda}\abs{g(k)}^2-\int_{|k|\leq \lambda}|g(k)|^2\dd k}\leq KL^{-1},
	$$
	uniformly in $\lambda\in\mathbb{R}$, for some $K>0$.
	Finally, note that the first inequality in \eqref{eq:SBdensity} follows directly from the definition of $d(\nu_L,\nu_\infty)$.
	\hspace*{\fill}$\square$\\

\noindent
We now take the size $L$ of the torus to infinity while simultaneously taking a weak-coupling limit, by setting $\mu_L=L^{-\beta}$ for some $\beta>0$. The following theorem identifies the parameter region for $\beta$ where spontaneous emission occurs.
\begin{theorem}
	\label{Theorem:ExpDecaySB}
	Let $g\in\text{L}^2(\mathbb{R}^3)$ satisfy Assumption $(R)$ and let $E_0>0$ be such that $$\Gamma(E_0)=\pi\norm{g\restriction_{S(E_0)}}_{\text{L}^2(S(E_0),\dd\omega_{E_0})}^2>0.$$ Fix a $\beta\in (0,1/2)\cup(1/2,\infty)$ and put $\mu_L=L^{-\beta}$. Then the simultaneous weak coupling and continuum limit
\begin{equation*}
	a_{wcc}(\tau):=\lim_{L\to\infty}\left\langle e\otimes \Omega, \exp(-i\tau\mu_L^{-2}(H_L^{\text{SB}}(E_0,\mu_Lg_L)-E_0)) e\otimes\Omega \right\rangle,\quad \tau\geq 0,
\end{equation*}
exists and is given by
\begin{equation}
	a_{wcc}(\tau)=
	\begin{cases}
		e^{i\Delta(E_0) \tau}e^{-\Gamma(E_0) \tau} & \text{if } \beta \in (0,1/2)\\
		1 & \text{if }  \beta \in (1/2,\infty).
	\end{cases}
\end{equation}
Moreover, there is a $C>0$ and, for any $0<\alpha_-<1/2$ and $\tau\geq 0$, a $C_{\tau,\alpha_-}>0$, such that
\begin{align*}
&\abs{a_{wcc}(\tau)-\left\langle e\otimes \Omega, \exp(-i\tau\mu_L^{-2}(H_L^{\text{SB}}(E_0,\mu_Lg_L)-E_0)) e\otimes\Omega \right\rangle}\\
&\quad\leq \begin{cases}
	C_{\tau,\alpha_-} L^{-\min\{2\beta \alpha_-,(1-2\beta)\alpha_-\}}& \text{if } \beta \in (0,1/2)\\
	CL^{1-2\beta} & \text{if }  \beta \in (1/2,\infty).
\end{cases}
\end{align*}
\end{theorem}
\noindent
\textit{Proof. }We shall prove the equivalent claim for the sequence of Friedrichs Hamiltonians $H_L(\mu_L)$ associated to the triples $(E_0,K_L,\mu_Lg_L)$. The case where $\beta\in(0,1/2)$ is a direct consequence of Corollary \ref{Corallary_SimWeakContLimit} and Lemma \ref{Lemm_Properties}. 
To prove the claim for $\beta\in(1/2,\infty)$, let $(\varepsilon_L)_{L\in\mathbb{N}}$ with $\varepsilon_L\in(1/2,3/2)$ be a sequence such that for sufficiently large $L\in\mathbb{N}$
$$\gamma(L)\equiv\set{z\in\mathbb{C}\text{ : }\abs{z-E_0}=\varepsilon_L}\subset\rho\left(H_{L}(\mu_L)\right),$$
and such that, for some $c>0$, $\text{dist}(z,\sigma(K_{L}))\geq cL^{-2}$ for  $z\in\gamma(L)$. Note that the first condition can always be achieved since $H_{L}(\mu_L)$ has discrete spectrum. The second condition can be achieved since $\sigma(K_{L})=\left\{\frac{2\pi |k|}{L}\text{ : }k\in\mathbb{Z}^3\right\}$, from which it follows that \begin{equation}
	\label{eq:specdist}
	\inf\left\{|\lambda_1-\lambda_2| : \lambda_1,\lambda_2\in \sigma(K_L)\cap[E_0-1,E_0+1], \lambda_1\neq \lambda_2\right\} \geq C_1L^{-2},
\end{equation}
for some constant $C_1>0$. To see this, let $\lambda_1=2\pi L^{-1}|k_1|\in \sigma(K_L)\cap [E_0-1,E_0+1]$, $k_1\in \mathbb{Z}^3$, be arbitrary. For any $\lambda_2\in \sigma(K_L)$, we then have that $$|\lambda_1-\lambda_2|\geq 2\pi L^{-1}\left(\sqrt{k_1^2+1}-|k_1|\right)=\frac{2\pi L^{-1}}{\sqrt{k_1^2+1}+|k_1|}\geq \frac{\pi}{L^2}\frac{1}{\sqrt{\left(\frac{E_0+1}{2\pi}\right)^2+ L^{-2}}}.$$
Thus, the estimate \eqref{eq:specdist} holds with $C_1=\pi\left[\left(\frac{E_0+1}{2\pi}\right)^2+1\right]^{-1/2}$. Next, we shall prove that there is a $C_2>0$ such that \begin{equation}
	\label{eq:SelfenergyboundL}
	\sup_{z\in\gamma(L)}\abs{\Sigma_L(z)}\leq C_2L.
\end{equation}
To that aim, note that for any $z\in \rho(K_L)$,
$$
\Sigma_L(z)=C_L^2\left(\frac{2\pi}{L}\right)^3\sum_{k\in\Lambda_L}\frac{|g(k)|^2}{|k|-z}.
$$
Now fix a $z\in \gamma(L)$. We separate the summation region into $\Lambda_L\equiv \frac{2\pi}{L}\mathbb{Z}^3=\Lambda_L^{(1)}\cup\Lambda_L^{(2)}$, where $\Lambda_L^{(1)}=\{k\in\Lambda_L\text{ : }\abs{|k|-z}> L^{-1}\}$ and $\Lambda_L^{(2)}=\{k\in\Lambda_L\text{ : }\abs{|k|-z}\leq L^{-1}\}$. For the summation over $\Lambda_L^{(1)}$, we have that
$$
\abs{C_L^2\left(\frac{2\pi}{L}\right)^3\sum_{k\in\Lambda_L^{(1)}}\frac{|g(k)|^2}{|k|-z}}\leq L^{-1} C_L^2\left(\frac{2\pi}{L}\right)^3\sum_{k\in\Lambda_L}|g(k)|^2=L^{-1},
$$
using the definition of $C_L$. By choice of the contour $\gamma(L)$, it follows that $\abs{|k|-z}\geq cL^{-2}$ for all $k\in\Lambda_L$. Hence, denoting by $\#$ the counting measure, we find
$$
\abs{C_L^2\left(\frac{2\pi}{L}\right)^3\sum_{k\in\Lambda_L^{(2)}}\frac{|g(k)|^2}{|k|-z}}\leq \frac{C_L^2(2\pi)^3}{cL}\norm{g}_{L^\infty(\mathbb{R}^3)}^2\#\Lambda_L^{(2)},
$$
so that the estimate \eqref{eq:SelfenergyboundL} follows if we can prove that there is a $C_3>0$ so that $\#\Lambda_L^{(2)}\leq C_3 L^2$. Let $\tilde{\Lambda}_L^{(2)}=\left\{k\in\Lambda_L\text{ : }\abs{|k|-|z|}\leq L^{-1}\right\}$ and put $Q_L^{(2)}=\bigcup_{k\in\tilde{\Lambda}_L^{(2)}}Q_L(k)$, where $Q_L(k)$ are the cubes of side-length $2\pi L^{-1}$ centred around $k$. Since $\Lambda_L^{(2)}\subset \tilde{\Lambda}_L^{(2)}$ and $$Q_L^{(2)}\subset \{k\in\mathbb{R}^3\text{ : } \abs{|k|-|z|}\leq (1+2\pi\sqrt{3})L^{-1}\},$$
we find that
 $$\#\Lambda_L^{(2)}\leq \#\tilde{\Lambda}_L^{(2)}=\left(\frac{L}{2\pi}\right)^3 \abs{Q_L^{(2)}}\leq \left(\frac{L}{2\pi}\right)^3 4\pi \left(\abs{z} + 1+2\pi\sqrt{3}\right)^2\frac{2+4\pi \sqrt{3}}{L},
 $$
 proving the estimate for $\#\Lambda_L^{(2)}$. Now, define the projection
$$
P_L=\frac{1}{2\pi i}\oint_{\gamma(L)}\left(H_{L}(\mu_L)-z\right)^{-1}\dd z.
$$
Since for any $z\in\gamma(L)$, $\abs{\Sigma_{L}(z)}\leq C_2L$, the Neumann series
$$
\left\langle\chi,\left(H_{L}(\mu_L)-z\right)^{-1}\chi\right\rangle=(E_0-z)^{-1}\sum_{m=0}^\infty \left[\mu_L^2\Sigma_{L}(z) (E_0-z)^{-1}\right]^m
$$
is absolutely convergent for all $z\in\gamma(L)$, when $L\in\mathbb{N}$ is sufficiently large. Using the representation
$$
\left\langle\chi,\exp(-itH_{L}(\mu_L))P_L\chi\right\rangle=\frac{1}{2\pi i}\oint_{\gamma(L)}e^{-itz}\left\langle\chi,\left(H_{L}(\mu_L)-z\right)^{-1}\chi\right\rangle\dd z,
$$
it follows that there is a $C_4>0$ such that
$$
\abs{\left\langle\chi,\exp(-itH_{L}(\mu_L))P_L\chi\right\rangle-e^{-itE_0}}\leq C_4L^{1-2\beta}.
$$
In particular, by setting $t=0$, this implies that $\norm{(1-P_L)\chi}\leq C_4L^{1-2\beta}$.
Thus, using the unitarity of $\exp(-itH_{L}(\mu_L))$, we have that
$$
\abs{\left\langle\chi,\exp(-itH_{L}(\mu_L))\chi\right\rangle-e^{-itE_0}}\leq 2C_4 L^{1-2\beta}.
$$
Since the estimate is uniform in $t\in\mathbb{R}$, the $\beta\in(1/2,\infty)$ case follows.
\hspace*{\fill}$\square$

\vspace{3mm}

\noindent\textbf{Acknowledgments.} Funded by the Deutsche
Forschungsgemeinschaft (DFG, German Research Foundation) – project
number 505496137.
\printbibliography
\end{document}